\documentclass[3p,times]{elsarticle}

\usepackage{ecrc}

\volume{00}

\firstpage{1}

\journalname{}

\runauth{Chen Wang et al.}

\jid{procs}

\jnltitlelogo{}

\CopyrightLine{2023}{Published by Elsevier Ltd.}

\usepackage{amssymb}
\usepackage{amsmath}
\usepackage{amsthm}
\usepackage{lineno,hyperref}
\usepackage{cases}
\usepackage{bm}
\usepackage[linesnumbered,ruled,vlined]{algorithm2e}
\usepackage{hyperref}
\usepackage[noend]{algpseudocode}
\modulolinenumbers[1]
\newtheorem{definition}{Definition}[section]
\newtheorem{theorem}{Theorem}
\newtheorem{proposition}{Proposition}

\newtheorem{example}{Example}
\newtheorem{corollary}{Corollary}

\newtheoremstyle{prestyle}
{0} % Space above
{\topsep} % Space below
{\itshape} % Body font
{} % Indent amount
{\bfseries} % Theorem head font
{.} % Punctuation after theorem head
{.5em} % Space after theorem head
{} % Theorem head spec (can be left empty, meaning `normal')
\theoremstyle{prestyle}

\theoremstyle{definition}	

\begin{document}

\begin{frontmatter}



\dochead{}
%% e.g. \dochead{17th International Conference on Dynamical Processes in Excited States of Solids}

\title{Two Graphs: Resolving the Periodic Reversibility of One-dimensional Finite Cellular Automata}


\author[1]{Chen Wang}
\ead{2120220677@mail.nankai.edu.cn}

\author[2]{Junchi Ma}
\ead{majunchi@mail.nankai.edu.cn}

\author[3]{Chao Wang\corref{cor1}}
\ead{wangchao@nankai.edu.cn}

\author[4]{Defu Lin}
\ead{2014580@mail.nankai.edu.cn}

\author[5]{Weilin Chen}
\ead{2120220681@mail.nankai.edu.cn}

\address[1,2,3,4,5]{Address: College of Software, Nankai University, Tianjin 300350, China}

\cortext[cor1]{Corresponding author.}

\begin{abstract}
	
Finite cellular automata (FCA) are widely used in simulating nonlinear complex systems, and their reversibility is closely related to information loss during the evolution. However, only a relatively small portion of their reversibility problems has been solved. In this paper, we perform calculations on two graphs and discover that the reversibility of any one-dimensional FCA exhibits periodicity as the number of cells increases. We also successfully provide a method to compute the reversibility sequence that encompasses the reversibility of one-dimensional FCA with any number of cells. Additionally, the calculations in this paper are applicable to FCA with various types of boundaries. This means that we will have an efficient method to determine the reversibility of almost all one-dimensional FCA, with a complexity independent of cell number. 
	
\end{abstract}


\end{frontmatter}


%\linenumbers

\section{Introduction}
\textbf{Cellular Automaton} (CA) is a discrete model, consisting of a rule and a grid of one or more dimensions. Each grid position has a state, and these states are updated in parallel at discrete time steps based on a fixed local rule. The concept of CA was formally introduced by John von Neumann in 1951 to simulate the self-replication of cells in biological development \cite{1951_Neumann}. CA has a wide range of applications in simulation, encryption, pseudorandom number generation, and more \cite{jun2009image,kang2008pseudorandom,rosin2014cellular}.

\textbf{Reversibility} in CA is one of the most important issues in the field. A reversible CA means that given its state at a certain time step, its state in the previous time step can be accurately determined. In other words, its evolution is reversible. Reversible CA can be used to simulate certain physical processes, such as time-reversal symmetry in quantum mechanics \cite{brun2020quantum}. In information processing, the ability to trace and recover previous states is useful, and a reversible CA offers this possibility \cite{toffoli1990invertible}. Reversibility can also be used in the design of encryption algorithms where the decryption process requires the recovery of the original state \cite{mohamed2014parallel}. The earliest study of reversibility was the concept of the ``Garden-of-Eden" configuration, first proposed by John Myhill in the 1960s and further studied by Moore \cite{1962_Moore,1963_Myhill}. A ``Garden-of-Eden" configuration is a special configuration that has no predecessor, implying a CA exhibiting the ``Garden-of-Eden" is not surjective. Amoroso and Patt in the 1970s discovered decision algorithms for injective and surjective CA and organized the relationship between injectivity and surjectivity \cite{1975_Amoroso,1972_Amoroso}, marking a milestone breakthrough. Bruckner proved that injective and surjective CA must be ``balanced" \cite{1979_Bruckner}. Sutner, based on the de Bruijn graph, provided another decision algorithm for injectivity and surjectivity and proved its complexity is quadratic time \cite{1991_Sutner}. In 1994, Kari proved that the reversibility of two-dimensional and higher-dimensional CA is undecidable, making it difficult to discuss the reversibility of high-dimensional CA \cite{1990_Kari,1994_Kari}. 
In recent years, Wolnik et al. have solved the reversibility problem of number-conserving CA using a very elegant method \cite{DZEDZEJ2020104534, WOLNIK2020180, WOLNIK2022133075}. However, the reversibility of one-dimensional CA has not been fully resolved.

Finite cellular automata (FCA) introduce boundaries to the earliest models of CA. After Kari, many researchers shifted the study of reversibility to linear FCA, using tools like matrices and polynomials. Del Rey made extensive attempts in this area: he analyzed the periodic relationship between the reversibility of a series of specific linear CA and the increase in the number of cells. This approach is feasible when there are few types of linear FCA \cite{del2009reversibility,MARTINDELREY20118360}. However, as the number of types increases and the analysis becomes complicated due to the enlargement of neighborhoods and the number of states, it becomes difficult to carry out. Building on this, Wang extended the computation to the general case of linear FCA, finding the periods of reversibility for any linear FCA through DFA and shift registers \cite{2022_ChaoWang,2015_ChaoWang}. This is already a significant breakthrough.

However, the complexity and functionality of linear FCA is much lower compared to nonlinear FCA. There are only $2^3=8$ elementary linear FCA, while there are a total of $2^{2^3}=256$ general ones. Therefore, only a very small part of the problem has been solved. Of course, there are several algorithms for determining the reversibility of general FCA, but these algorithms have their limitations: not only are they overly complex and unable to determine the period, but they are also only applicable to 3-neighborhood-size FCA and are suitable for only one type of boundary, either null boundary or periodic boundary \cite{2015_Bhattacharjee,2010_Maiti}.

In this paper, we perform calculations on two new kinds of graphs to discover that the reversibility of any nonlinear FCA shows periodicity as cell number $n$ grows. We successfully provide a method to calculate a sequence (reversibility sequence), which contains the reversibility of FCA with all numbers of cells. This means we will have an efficient way to determine the reversibility of any FCA completely. The complexity of this algorithm is independent of $n$, a significant leap compared to Maiti's algorithm \cite{2010_Maiti}. This algorithm is both applicable to null boundaries and periodic boundaries. Our paper demonstrates that, \textbf{even for complex one-dimensional nonlinear FCA, the computation of their reversibility is efficient}.

This paper consists of four sections. The second section introduces the FCA and related concepts. The core content of the algorithm is in the third section. The fourth section summarizes the work of the entire paper.

\section{Preliminaries \label{s2}}
In this section, we give the formal definitions used in the article. We first define the finite cellular automata (FCA) and relevant problems. Then we define different boundaries of one-dimensional FCA. 

\subsection{Finite cellular automata \label{s21}}
Finite cellular automaton (FCA) is the most basic and important concept in this paper. It is usually defined by a quintuple: FCA = $\{Z^d,S,N_m,f,b\}$
\begin{itemize}
	\item $Z^d$ is the dimension of the cellular space. Each point $c \in Z^d$ is called a cell. In this paper, the dimension $d=1$. The number of cells in $Z^1$ is denoted as $n$ in this paper.
	\item $S=\{0,1,\ldots, s-1\}$ is a finite set of states representing a cell's state in CA. At any time, the state of each cell is an element in this set. We use $s$ to count the number of elements in this set.
	\item $N_m=(\vec n_1, \vec n_2, \ldots , \vec n_m)$ represents neighbor vectors, where $\vec n_i \in \mathbb{Z}^d$, and $\vec n_i \neq \vec n_j$ when $i \neq j$ $(i, j = 1, 2,\ldots, m)$. Thus, the neighbors of the cell $\vec n \in \mathbb{Z}^d$ are the $m$ cells $\vec n + \vec n_i, i = 1, 2,\cdots, m$. The neighborhood size of this CA is denoted as $m$ in this paper. We use $l+1+r$ to express a function containing $l$ neighbor on the left and $r$ neighbor on the right. For example, the rule of an elementary FCA can be expressed as $1+1+1$.
	\item $f$: $S^m \rightarrow S$ is the local rule (or simply the rule). The rule maps the current state of a cell and all its neighbors to this cell's next state. 
	\item $b$ is the type of boundary which includes the null boundary ($b=``n"$) and the periodic boundary ($b=``p"$). The algorithm in this paper is also applicable to the reflective boundary, but it will not be further elaborated upon. 
\end{itemize}

A configuration is a mapping $C$ : $Z^d \rightarrow S$ which assigns each cell a state. Function $\tau$ : $C \rightarrow C$ is used to represent the global transformation of $f$.

The Wolfram number is a shorthand for the rule. For instance, for a $FCA = \{Z^1,\{0,1\},N_3,f,b\}$ the rule $f$ is shown as follows:
\begin{equation}
	\begin{split}
		f(0,0,0)=0, \ \ \ \  f(0,0,1)=0,\ \ \ \  f(0,1,0)=0, \ \ \ \ f(0,1,1)=0,\\
		f(1,0,0)=1, \ \ \ \  f(1,0,1)=1,\ \ \ \  f(1,1,0)=1, \ \ \ \ f(1,1,1)=1.
	\end{split}
\end{equation}
Then, we can use the Wolfram number 11110000 to abbreviate this rule. This FCA can be represented as a quintuple $ \{Z^1,\{0,1\},N_3,11110000,b\}$. The length of the Wolfram number is $s^m$.

\subsection{Boundaries of one-dimensional FCA \label{s22}}
There are several different boundaries of one-dimensional FCA. The most important ones studied at present are the null (fixed) boundary and the periodic boundary. Now, we will introduce these boundaries respectively.

Now we suppose the rule of an FCA is $1+1+1$, Then the leftmost cell has no left neighbor, and the rightmost cell has no right neighbor. In this case, the leftmost and rightmost cells in the next transformation will lose their corresponding value. To avoid this, we need to fill their neighbors outside the boundary with certain values. These two boundaries correspond to different values of their neighbors.

The null boundary is also called the zero boundary, which always fills ``$0$" in the cell outside the configuration. The example of a null boundary is shown in Fig. \ref{Fig2-1}.
\begin{figure}[h]
	\center
	\includegraphics[width=0.8\linewidth]{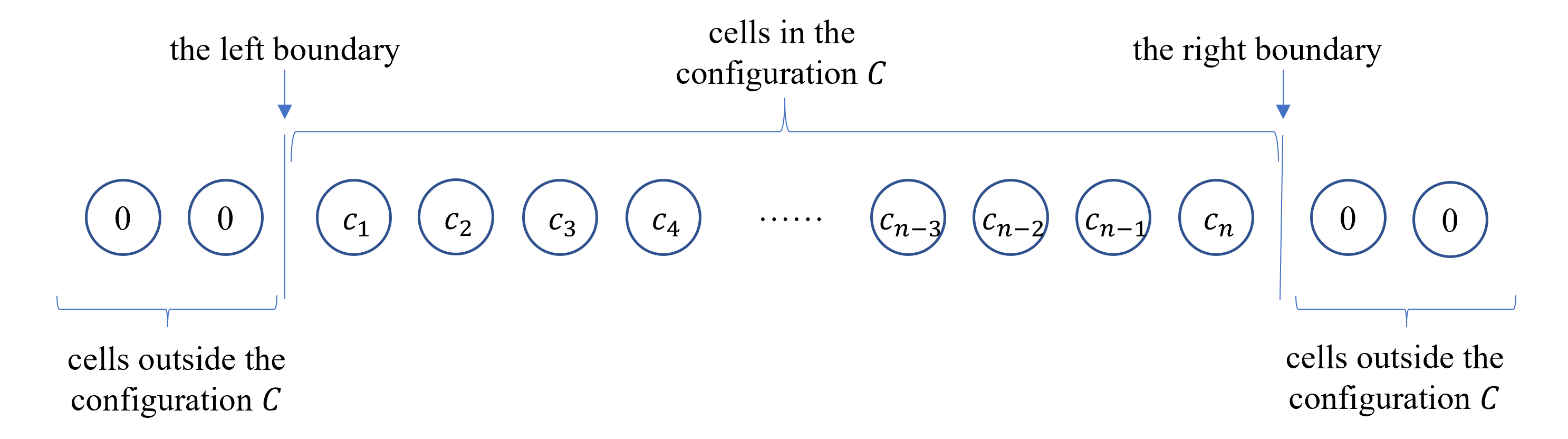}
	\caption{The null boundary of one-dimensional FCA}
	\label{Fig2-1}
\end{figure}

Periodic boundary regards the configuration as a circle, which fills the cells left-outside the configuration according to the rightmost cells. The example of the periodic boundary is also shown in Fig. \ref{Fig2-2}.

\begin{figure}[h]
	\center
	\includegraphics[width=0.8\linewidth]{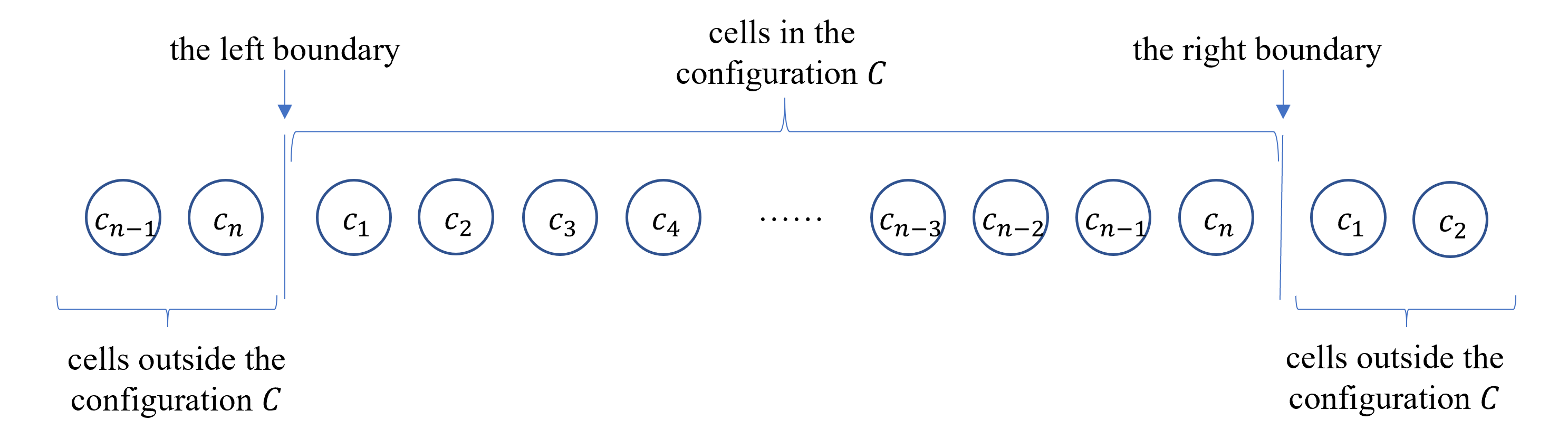}
	\caption{The periodic boundary of one-dimensional FCA}
	\label{Fig2-2}
\end{figure}

\subsection{Reversibility problem \label{s23}}
\begin{definition}
	The \textbf{surjectivity} of a $FCA = \{Z^1,S,N_m,f,b\}$ can be expressed as, for any configuration $c_1 \in C$, there exists a configuration $c_2\in C$, we have $\tau(c_2) = c_1$. This means that every configuration can find its predecessor, indicating that there is no ``Garden-of-Eden" configuration.
\end{definition}

\begin{definition}
	The \textbf{injectivity} of a $FCA = \{Z^1, S, N_m, f, b\}$ can be expressed similarly: for any two configurations $c_1, c_2 \in C$, if $\tau(c_1)=\tau(c_2)$, then there must be $c_1=c_2$. That means all configurations have one and only one predecessor. 
\end{definition}

\begin{proposition}
	For a FCA, surjective and injective are equivalent and are also referred to as bijective or reversible.
\end{proposition}
\begin{proof}
	For a FCA, since the number of cells is fixed, the two sets before and after the mapping by the global mapping function $\tau$ have an equal number of elements. Therefore, the surjectivity is equivalent to injectivity, which is a bijectivity, also known as reversibility.
\end{proof}
\begin{definition}
	If all $n$-cell-configurations (the number is $s^n$) of a $FCA = \{Z^1,S,N_m,f,b\}$ has one and only one predecessor, then this FCA is \textbf{$n$-cell-reversible}.
\end{definition}

\begin{definition}
	If for any $n$, the FCA is always $n$-cell-reversible, then this FCA is strictly reversible. Conversely, if it is never $n$-cell-reversible, then it is strictly irreversible.
\end{definition}

\begin{definition}
	The \textbf{reversibility sequence} is an infinite sequence of 0s and 1s, where the $n$th number from the left is 0 indicates that the FCA is irreversible when it has $n$ cells, and 1 indicates its reversibility. Due to the presence of periodicity, we underline a complete period in the sequence. For example, if an FCA's reversibility sequence is 010\underline{001}, then this FCA is reversible only when the number of cells is in $\{2\} \cup \{3k\ |\ k>1$ and $k \in \mathbb{Z_+}\}$. Its reversibility period is 3.
\end{definition}

Our goal is to quickly determine whether any given FCA is strictly reversible or strictly irreversible. If it is neither, we should obtain its reversibility sequence directly by calculating the period. Specifically, for various widely used boundary conditions, the algorithm should be fully compatible (without significant differences).

\subsection{Amoroso's algorithm \label{s24}}
Amoroso's algorithm is used to determine the surjectivity of infinite CA. It can be described in brief step by step.
\begin{description}
	\item[Step 1] For a CA with neighborhood size $m$ and state set $S=\{0,1,\ldots, p-1\}$, get its local rule $f$.
	\item[Step 2] If $f(m)=0$, add $m$ into the root node.
	\item[Step 3] For each node $M$ in layer $i$ $(i \geq 0)$, construct its children $M_0, M_1,\cdots, M_{p-1}$. For each tuple $a_1a_2 \cdots a_{m}$ in $M$ and all $0 \leq j<p$, if $f(a_2a_3 \cdots a_mj)=t$, then $a_2a_3 \cdots a_mj$ is added to $M_t$. If a node is the same as a node that has appeared before, its children will not be constructed. If there are not any m-tuples in a node, then the CA has a Garden-of-Eden which decides the CA is not surjective.
	\item[Step 4] If the three steps above are completed and the CA is not decided as non-surjective, then the CA is surjective.
\end{description}

%\subsection{The finite surjectivity tree \label{s23}}
%Wang et al. detailed in their paper how to determine whether an FCA is reversible for any configuration length, including %null boundary, periodic boundary, and reflective boundary \cite{wang2023stateofart}. The construction of surjection trees is very useful in this paper, but due to the length, we only briefly the surjectivity tree with periodic boundary here.
%\begin{description}
%	\item[Step 1] Input a FCA $ = \{Z^1,S,N_m,f,``p"\}$.
%	\item[Step 2] Set $N_{ini} = \{(\bigcup_{i=1}^{m-1} a_i, \bigcup_{i=1}^{m-1} b_i)| a_i = b_i, a_i, b_i\in S\}$ as root of the tree.
%	\item[Step 3] For each new node $N_{cur}$ in layer $i$ $(i \geq 0)$, construct its children $N_{child(0)}, N_{child(1)},\cdots, N_{child(p-1)}$ where $N_{child(c)} = \{(\bigcup_{i=1}^{m-1} a_i, \bigcup_{i=1}^{m-1} b_i)| \exists b_0 \in S, (\bigcup_{i=1}^{m-1} a_i, \bigcup_{i=0}^{m-2} b_i) \in N_{cur}$ and $f(\bigcup_{i=0}^{m-1} b_i) = c, a_i, b_i\in S\}$ $(c\in S)$. If there is a negative node which $N_{cur} \cap N_{ini} = \emptyset$, this FCA is not injective. If there are not any negative nodes, it is injective.
%\end{description}
%\begin{proposition} \cite{wang2023stateofart}
%	The complexity of the finite surjectivity tree of a FCA is $O(s^{2m})$
%\end{proposition}
%We can find this algorithm is so simple and efficient. It can also determine whether a CA is injective.

\section{Main result \label{s3}}

The main process of the algorithm is divided into three steps: The first step is to construct the reversibility graph. The second step is to find all circuits in the reversibility graph and construct the circuit graph. The third step is to list the expressions of all irreversible vertices and ultimately provide the reversibility period and sequence. %In this article, the ``circuit" is assumed to refer to ``elementary circuit" by default. 

\subsection{Reversibility graph \label{ss31}}

\begin{definition}
	For FCA = $\{Z^1,S,N_m,f,``n"\}$ $(m=l+1+r)$, its reversibility graph $G_R=(V,E)$ is defined as follows:
	\begin{itemize}
		\item $V=\{v : v \subseteq S^{m-1},\text{ and there is a path from $v_{root}$ to $v$}\}$.
		\item $v_{root} = \{a_1,\cdots,a_{m-1}:a_1,\cdots,a_{m-1} \in S^{m-1} \text{ and } \forall i \in \{1,\cdots,l\},a_i=0\}$.
		\item $v_{end} = \{a_1,\cdots,a_{m-1}:a_1,\cdots,a_{m-1} \in S^{m-1} \text{ and } \forall i \in \{l+1,\cdots,m-1\},a_i=0\}$.
		\item For $u,v \in V$, there is an edge $(u,v)$ labeled $e$ from $u$ to $v$ iff $\forall a=\{a_1,\cdots,a_{m-1}\} \in v$ and $i\in \{1,\cdots,m-2\}$, $ \exists b=\{b_1,\cdots,b_{m-1}\} \in u$, there are $a_i=b_{i+1}$ and $f(b_1,a_1,\cdots,a_{m-1})=e$
	\end{itemize}
\end{definition}

\begin{definition}
	For FCA = $\{Z^1,S,N_m,f,``p"\}$ $(m=l+1+r)$, its reversibility graph $G_R=(V,E)$ is defined as follows:
	\begin{itemize}
		\item $V=\{v : v \subseteq S^{2m-2},\text{and there is a path from $v_{root}$ to $v$}\}$.
		\item $v_{root} = v_{end} = \{a_1,\cdots,a_{2m-2}:a_1,\cdots,a_{2m-2} \in S^{2m-2} \text{ and } \forall i \in \{1,\cdots,m-1\},a_i=a_{i+m-1}\}$.
		\item For $u,v \in V$, there is an edge $(u,v)$ labeled $e$ from $u$ to $v$ iff $\forall a=\{a_1,\cdots,a_{m-1}\} \in v$, $i\in \{1,\cdots,m-2\}$, and $j\in \{1,\cdots,m-1\}$, $ \exists b=\{b_1,\cdots,b_{m-1}\} \in u$, there are $a_j=b_j$, $a_{i+m-1}=b_{i+m}$ and $f(b_m,a_m,\cdots,a_{2m-2})=e$
	\end{itemize}
\end{definition}

\begin{definition}
	The negative vertex set $V_N = \{v \in V:v \cap end \ne \emptyset\}$.
\end{definition}

\begin{theorem}
	FCA = $\{Z^1,S,N_m,f, b\}$ is strictly reversible iff $V_N = \emptyset$ in $G_R$.
\end{theorem}
\begin{proof}
	$G_R$ is very similar to the tree constructed by Amoroso's algorithm introduced in Subsection \ref{s24}. The difference lies in the removal of duplicate vertices from the tree and the addition of $v_{root}$ and $v_{end}$ to accommodate the boundary conditions of FCA, thus the proofs are also similar.
	
	If there are no negative vertices in $G_R$, for any configuration $c$, there exists a corresponding path starting from $v_{root}$, and among the vertices the path traverses, there must be a valid configuration $c'$ that satisfies $\tau(c')=c$. Since any configurations can find predecessors through the vertices on the path, this FCA is strictly reversible. 
	Conversely, if there is at least one negative vertex in $G_R$, the corresponding configurations have no predecessors and are Garden-of-Eden configurations, meaning that the FCA is not strictly reversible. 
\end{proof}

\begin{proposition}
	$G_R$ can be constructed within $O(s^mV)$.
\end{proposition}
\begin{proof}
	For FCA, the complexity of the rule $f$ is $s^m$. Therefore, the complexity of $G_R$ is linearly related to both $V$ and the rule complexity $s^m$, resulting in an overall complexity of $O(s^mV)$. Fig. \ref{RG} and Algorithm \ref{a} shows the structure and a detailed pseudocode of $G_R$.
\end{proof}
\begin{figure}[h]
	\center
	\includegraphics[width=0.48\linewidth]{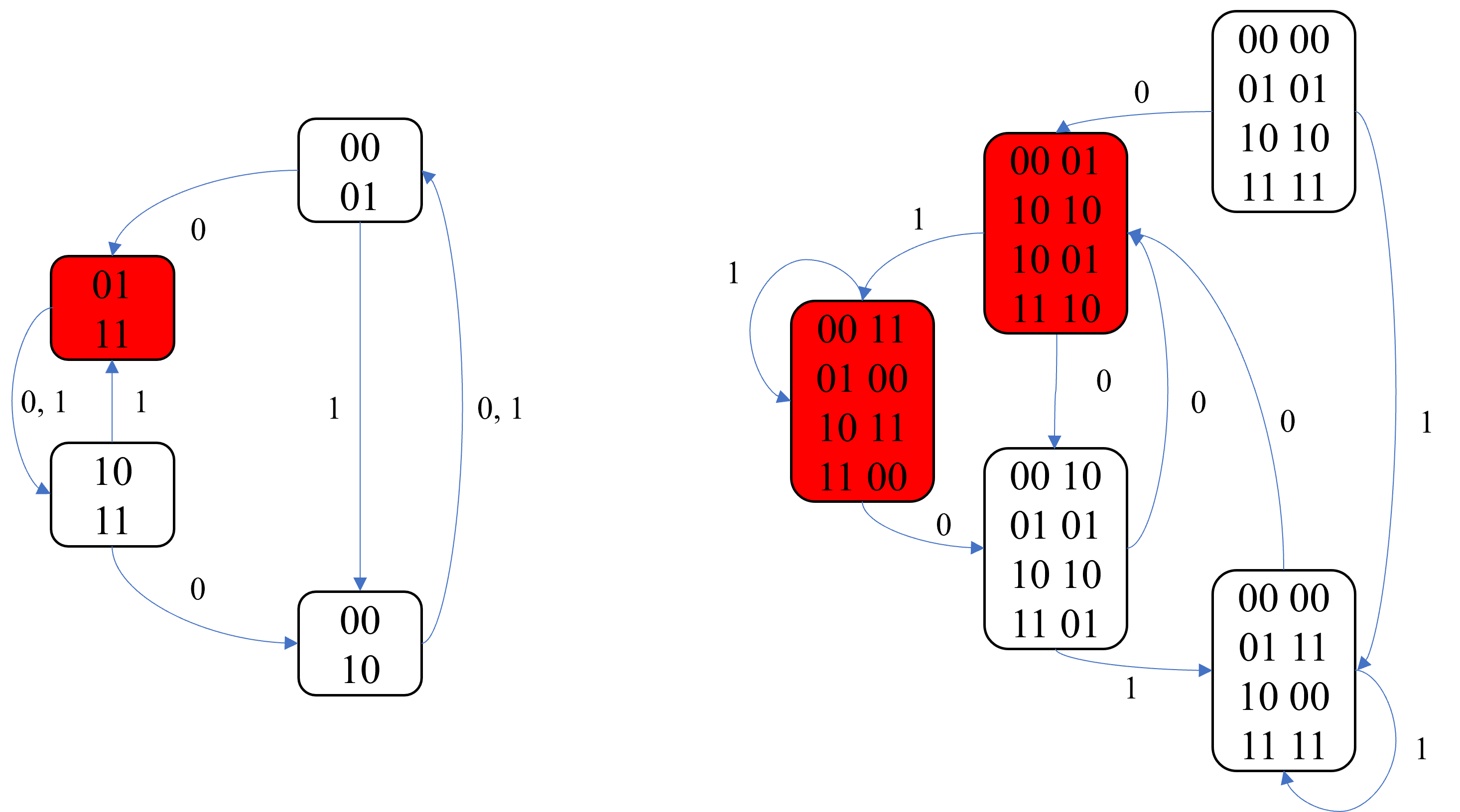}
	\caption{The left side is the reversibility graph for FCA=$\{Z^1,\{0,1\},N_3,10100101,``n"\}$, while the right side is for FCA = $\{Z^1,\{0,1\},N_3,10011001,``p"\}$. Negative vertices are marked in red.}
	\label{RG}
\end{figure}
\begin{algorithm}[H]
	\small
	\SetAlgoLined
	\label{a}
	\KwData{one-dimensional finite CA $\{Z^1, S, N_m, f, b\} (m = l + 1 + r)$}
	\KwResult{reversibility graph $G_R$}
	let $G_R$ be an empty directed graph, $Q$ be an empty queue\;
	\eIf{b==``n"}{
		$v_{root} = \{a_1a_2\cdots a_{m-1}|a_1a_2\cdots a_{m-1}\in S^{m-1}, a_j=0, 1\leq j\leq l\}$\;
		$v_{end} = \{a_1a_2\cdots a_{m-1}|a_1a_2\cdots a_{m-1}\in S^{m-1}, a_j=0, m - r\leq j\leq m-1\}$\;
	}(b==``p"){
		$v_{root} = v_{end} = \{a_1a_2\cdots a_{m-1}b_1b_2\cdots b_{m-1}|a_1a_2\cdots a_{m-1}b_1b_2\cdots b_{m-1} \in S^{2m-2}, a_i = b_i, 1\leq i\leq m-1\}$\;
	}
	$d(v_{root}) = 0$\;
	add $v_{root}$ to $G_R$ and $Q$\; 
	\While{$Q$ is not empty}{
		pop the first vertex in $Q$ and denote it with $v_{cur}$\;
		\If{$v_{cur} \cap v_{end} == \emptyset$}{
			mark $v_{cur}$ red\;
		}
		\For{each $e \in S$}{	
			\eIf{b==``n"}{
				$v_e = \{a_1a_2\cdots a_{m-1}|a_1a_2\cdots a_{m-1}\in S^{m-1}, \exists a_0 \in S, a_0a_1\cdots a_{m-2} \in v_{cur}$ and $f(a_0a_1\cdots a_{m-1}) = e\}$\;
			}(b==``p"){
				$v_e = \{a_1a_2\cdots a_{m-1}b_1b_2\cdots b_{m-1}|a_1a_2\cdots a_{m-1}b_1b_2\cdots b_{m-1} \in S^{2m-2}, \exists b_0 \in S, a_1a_2\cdots a_{m-1}b_0b_1\cdots b_{m-2} \in v_{cur}$ and $f(b_0b_1\cdots b_{m-1}) = e \}$\;
			}
			\eIf{$\exists v_{same} == v_e$ in $G_R$}{
				add an edge $(v_{cur},v_{same})$ labeled $e$ into $G_R$\;
			}{
				$d(v_e) = d(v_{cur})+1$ \;
				add an edge $(v_{cur},v_{e})$ labeled $e$ into $G_R$\;
				add $v_e$ to $G_R$ and $Q$\;
			}		
		}
	}
	
	return $G_R$\;
	\caption{the reversibility graph $G_R$}
\end{algorithm}

\subsection{Circuit graph and reversibility period}
In the previous section, we introduce the construction of the reversibility graph $G_R$ and use it to determine if an FCA is strictly reversible. For FCA with a specific number of cells, $G_R$ remains available, but the concept of ``depth" needs to be introduced.
\begin{definition}
	The depth of $v \in V$ in $G_R$ is $d(v)=min\{d(u)+1:(u,v)\in E\}$ ($d(v_{root})=0$). 
\end{definition}
\begin{figure}[h]
	\center
	\includegraphics[width=0.5\linewidth]{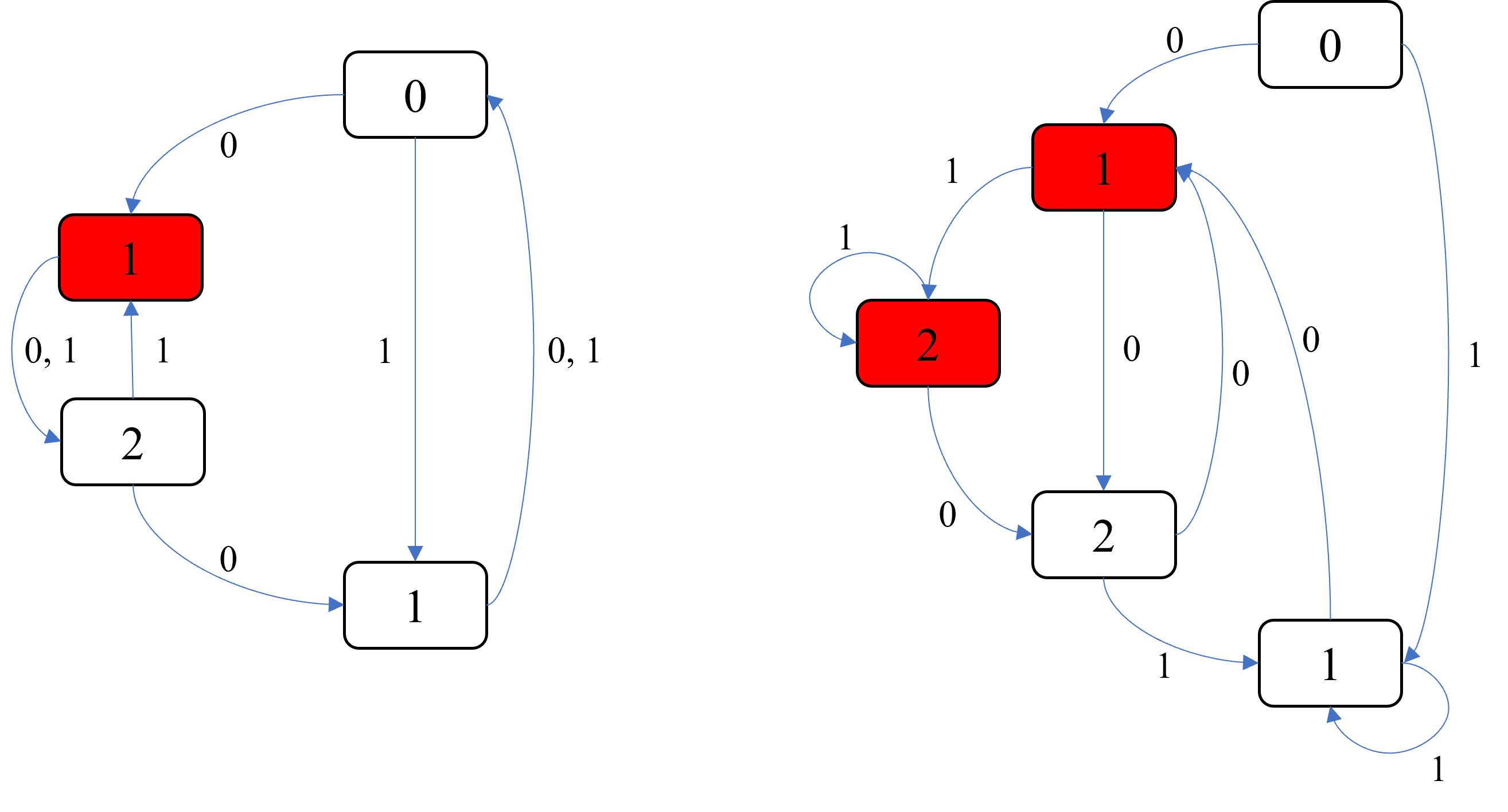}
	\caption{Replace the content of each vertex $v$ of $G_R$ in Fig. \ref{RG} with its depth $d(v)$.}
	\label{depth}
\end{figure}

\begin{theorem}
	If $v \in V_N$, then the FCA is irreversible when it has $d(v)$ cells.
\end{theorem}
\begin{proof}
	There is always a path from $v_{root}$ to $v$. The configuration corresponding to this path cannot find predecessors, indicating the existence of a ``Garden of Eden" configuration. For example, in Fig. \ref{RG} left, the configuration with one cell, 0, is a ``Garden-of-Eden" configuration.
\end{proof}

\begin{theorem}
	If $v \in V_N$, and $v$ is included in $i$ elementary circuits of length $r_1,r_2,\cdots,r_i$ . then the FCA is irreversible when it has $d(v)+x_1r_1+x_2r_2+\cdots+x_ir_i$ cells $(x_1,x_2,\cdots,x_i \in \mathbb{N}_0)$.
\end{theorem}
\begin{proof}
	Starting from $v$, there are $i$ different paths that return to $v$, and continuing the search, there are $i$ different paths again, essentially forming a linear combination of $r_1,r_2,\cdots,r_i$.
\end{proof}

The aforementioned theorems are still not sufficient. We can find that in Fig. \ref{RG} left, 000(01)*10 is also a ``Garden-of-Eden" configuration, yet we have not included it in our calculations. To address this issue, we define a new type of graph called the "circuit graph."

\begin{definition}
	For FCA = $\{Z^1,S,N_m,f,b\}$, its circuit graph $G_C=(V_C,E_C)$ is defined as follows:
	\begin{itemize}
		\item $G_C$ is an undirected graph.
		\item $V_C = V_{E} \cup V_N $, where $V_E$ is the set of elementary circuit in $G_R$.
		\item For $u,v \in V_C$, there is $(u,v) \in E_C$ iff $u$ and $v$ include at least one same vertex.
		\item For $u \in V_N$, $v \in V_C$, there is $(u,v) \in E_C$ iff $u$ is included in $v$.
		\item The $d(u)$ of $u \in V_N$ is its depth in $G_R$ and $l(v)$ of $v \in V_C$ is its length in $G_R$.
	\end{itemize}
\end{definition}

\begin{theorem}
	If  $v \in V_N$, $v_1,\cdots,v_i \in V_E$ and they are connected to $v$ in $G_C$, then the FCA is irreversible when it has $k+x_1l(v_1)+x_2l(v_2)+\cdots+x_il(v_i)$ cells where $x_1,x_2,\cdots,x_i \in \mathbb{N}_0$ and the subgraph consisting of $v_j$ where $x_j \ne 0$ and their related edges is a connected graph.
\end{theorem}
\begin{example}
	\label{example}
	The proof of this theorem is similar to that of the previous theorems, so we use an example in place of a formal proof here. Fig. \ref{CG} provides $G_C$ corresponding to $G_R$ in Fig. \ref{RG} right. For convenience of representation, we assign an order number to each vertex in the top left corner.
	\begin{figure}[h]
		\center
		\includegraphics[width=0.4\linewidth]{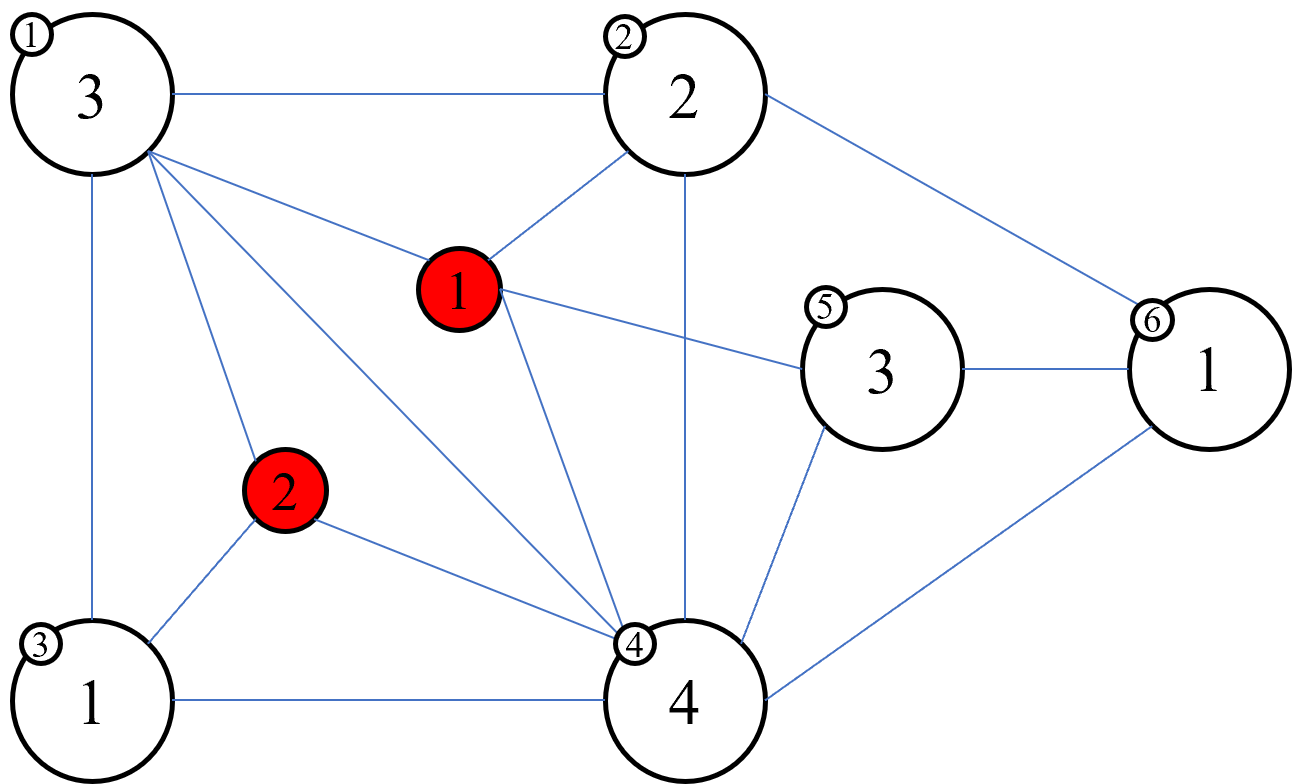}
		\caption{The circuit graph $G_C$ of $FCA = \{Z^1,\{0,1\},N_3,10011001,``p"\}$.}
		\label{CG}
	\end{figure}
	
	We can quickly list the irreversible number for the FCA based on the above theorems. 
	\begin{equation}
		\left\{\begin{matrix}
			1+3x_1+2x_2+x_3+4x_4+3x_5+x_6 \ \ \\
			2+3y_1+2y_2+y_3+4y_4+3y_5+y_6 \ \  
		\end{matrix}\right.
	\end{equation}
	
	From $G_C$, we know that $y_1=y_3=y_4=0$ is a sufficient condition for $y_5=y_6=0$. Many similar conditions are all easy to discern, so we will not list them all here. The union of all these irreversible number forms the reversibility sequence of the FCA. Ultimately, we can determine that the reversibility sequence for this FCA is: \underline{0}.
\end{example}

\begin{definition}
	For a connected component $G_{sub}=(V_{sub},E_{sub})$ of $G_C$, the sub-period is $\gcd(V_{sub} \cap V_C)$, where the function $\gcd$ is to find the maximum common divisor of the lengths of all circuits in set. 
\end{definition}

\begin{theorem}
	The reversibility period of FCA is $\mathrm{lcm}(P)$, where $P$ is the set of sub-periods and function $\mathrm{lcm}$ is to get the minimum common multiple of a set, then . 
\end{theorem}
\begin{proof}	
	The proof of this theorem stems from the Frobenius Coin Problem \cite{bardomero2020frobenius,shallit2008frobenius}. This problem can be briefly described as determining the largest amount of money that cannot be made with given denominations of coins. The problem has been widely explored in number theory and combinatorial mathematics. In this paper, we utilize it to calculate the reversibility period.
\end{proof}

\begin{corollary}
	\label{c1}
	If $v \in V_N$, $\exists u \in V_E$ in $G_C$ and $l(u)=1$, then the FCA is irreversible when the number $n$ of cells satisfies
	\begin{equation}
		n \ge d(v)+\sum_{w \in V_E}l(w).
	\end{equation}
\end{corollary}
\begin{corollary}
	\label{c2}
	If $v \in V_N$, $\exists u_1,u_2 \in V_E$ in $G_C$ and $\gcd(l(u_1),l(u_2))=1$, then the number $n$ of cells satisfies
	\begin{equation}
		n>d(v)+l_1l_2+\sum_{w \in V_E \setminus \{u_1,u_2\}}l(w).
	\end{equation}
\end{corollary}
\begin{proof}
	Corollary \ref{c1} is obvious, while the proof of Corollary \ref{c2} comes from the precise conclusion of the Frobenius Coin Problem.
\end{proof}
These two corollaries can help us make reversibility determinations more quickly and directly, sometimes even eliminating the need to construct $G_C$. In Example \ref{example}, we can determine that the FCA is irreversible when the number of cells is greater than or equal to 2, simply by finding there is a negative vertex in $G_R$ with a self-circuit.

Through the aforementioned analysis, we have resolved the most significant issues. Only two minor problems remain to be addressed: one is finding all the circuits in $G_R$, and the other is the construction of $G_C$.

\begin{proposition}
	\cite{johnson1975finding} We can find all elementary circuits of a directed graph in time bounded by $O((V + E)(l + 1))$, where there are $V$ vertices, $E$ edges and $l$ circuits in the graph. 
\end{proposition}
We can find all the circuits of $G_R$ in $O(sVl)$, where the FCA has $s$ states and $G_R$ has $V$ vertices and $l$ circuits. 

\begin{proposition}
	We can construct $G_C$ in time bounded by $O(Vl^2)$, where there are $V$ vertices and $l$ circuits in $G_R$. 
\end{proposition}
\begin{proof}		
	In $G_C$, there are $l$ points, so there are at most $l^2$ edges. When determining whether two points are adjacent, we need to check whether at most $V$ vertices appear simultaneously in these two circuits. Therefore, the complexity is $O(Vl^2)$. Algorithm \ref{b} shows the  detailed pseudocode of $G_C$.
\end{proof}

\begin{algorithm}[h]
	\small
	\SetAlgoLined
	\label{b}
	\KwData{a reversibility graph $G_R=(V,E)$}
	\KwResult{the circuit graph $G_C$}
	let $G_C=(V_C,E_C)$ be an empty directed graph\;
	Get the negative vertex set $V_N$ in $G_R$\;
	Calculate the elementary circuit set$ V_E$ in $G_R$ by algorithm in \cite{johnson1975finding}\;
	$V_C=V_N \cup V_E$\;
	\For{each $v,v' \in V_E$}{
		\If{$u$ and $v$ include at least one same vertex}{
			add $(v,v')$ into $E_C$\;
		}
	}
	\For{each $v \in V_N,v' \in V_E$}{
		\If{$v$ is included in $v'$}{
			add $(v,v')$ into $E_C$\;
		}
	}
	return $G_C=(V_C,E_C)$\;
	\caption{the circuit graph $G_C$}
\end{algorithm}

Thus, the reversibility problem of one-dimensional FCA has been efficiently resolved. Actually, determining whether an FCA with $n$ cells is reversible is not straightforward. However, we approach this problem from a novel perspective: we list all the irreversible scenarios, and what remains are naturally reversible, as a substitute for individually determining the reversibility for every cell count. This approach significantly simplifies the complexity of analyzing the entire problem. Next, let's fully recapitulate our algorithm.

\begin{description}
	\item[Step 1] Construct $G_R$ according to Algorithm \ref{a} with the complexity of $O(s^mV)$.
	\item[Step 2] Find all circuits using Johnson's algorithm with the complexity of $O(sVl)$ \cite{johnson1975finding}.
	\item[Step 3] Construct $G_C$ according to Algorithm \ref{b} with the complexity of $O(Vl^2)$.
	\item[Step 4] List all irreversible number polynomials and solve the reversibility sequence.
\end{description}

We list several reversibility sequences for a series of rules in Table \ref{table1}. Some of the rules in this chapter are linear and can be compared with existing results for linear FCA to verify the correctness of our algorithm. The underlined portions represent a period. $s$ is the number of states; $m$ is the neighborhood size; $b$ is the boundary and $f$ is the rule.

\begin{table}[h]
	\center
	\caption{List of reversibility sequence}
	\label{table1}
	\begin{tabular}{ccc}
		\hline
		$s$ , $m$ , $b$ & $f$ & reversibility sequence\\
		\hline
		&  10101010 & 00000000000$\underline{0}$\\
		&  10101001 & 10000000000$\underline{0}$\\
		2 , 3 , ``n" & 10100110 & 10000000000$\underline{0}$\\
		& 10100101 & 0101010101$\underline{01}$\\
		& 10010110 & 101101101$\underline{101}$\\
		\hline
		& 11100001 & 00000000000$\underline{0}$\\
		& 11010010 & 1010101010$\underline{10}$\\
		& 11000011 & 00000000000$\underline{0}$\\
		2 , 3 , ``p" & 10101001 & 00000000000$\underline{0}$\\
		&  10100110 & 1010101010$\underline{10}$\\
		&  10010110 & 110110110$\underline{110}$\\
		\hline
		&  0010101111010100 & 10110100000$\underline{0}$\\
		2 , 4 , ``n" & 0011001100111100 & 10000000000$\underline{0}$\\
		&  0011101011000011 & 10100000000$\underline{0}$\\
		&  1001011001101001 & 10011001$\underline{1001}$\\
		\hline
		&  022211222000011200010212111 & 10000000000$\underline{0}$\\
		3 , 3 , ``n" & 110002222110011222010022011 & 10000000000$\underline{0}$\\
		&  012012012120120120201201201 & 0101010101$\underline{01}$\\
		&  012120201120201012201012120 & 11011101$\underline{1101}$\\
		\hline 
	\end{tabular}
\end{table}

During the experimental process, we found that a large number of nonlinear FCA have a period of 1, meaning that starting from a certain length, configurations of greater lengths are all non-reversible. However, there is still a small portion of nonlinear FCA whose reversibility exhibits periodicity (e.g. the rule 11010010).

\section{Conclusion \label{s4}}
Our algorithm allows for the computation of the reversibility of configurations of any length for any FCA, with the complexity of $O(s^mV+sVl+Vl^2)$. Moreover, the characteristic of finding periods makes this algorithm particularly suitable for computing configurations of greater lengths. However, there are still some issues that remain. First, is there still room for optimization in this algorithm, to further reduce unnecessary computations? Second, during our experiments, we found that there are very few non-linear CA with a period larger than 1; is there a sufficient and necessary condition to identify these CA?

\section*{Acknowledgments}
This study is financed by Tianjin Science and Technology Bureau, finance code:21JCYBJC00210.

\section*{Declaration of generative AI}
Generative AI is only used for translation and language polishing in this paper.

\bibliographystyle{elsarticle-harv}
\bibliography{fbref}

\begin{thebibliography}{27}
\expandafter\ifx\csname natexlab\endcsname\relax\def\natexlab#1{#1}\fi
\providecommand{\url}[1]{\texttt{#1}}
\providecommand{\href}[2]{#2}
\providecommand{\path}[1]{#1}
\providecommand{\DOIprefix}{doi:}
\providecommand{\ArXivprefix}{arXiv:}
\providecommand{\URLprefix}{URL: }
\providecommand{\Pubmedprefix}{pmid:}
\providecommand{\doi}[1]{\href{http://dx.doi.org/#1}{\path{#1}}}
\providecommand{\Pubmed}[1]{\href{pmid:#1}{\path{#1}}}
\providecommand{\bibinfo}[2]{#2}
\ifx\xfnm\relax \def\xfnm[#1]{\unskip,\space#1}\fi
%Type = Article
\bibitem[{Amoroso et~al.(1975)Amoroso, Cooper and Patt}]{1975_Amoroso}
\bibinfo{author}{Amoroso, S.}, \bibinfo{author}{Cooper, G.},
  \bibinfo{author}{Patt, Y.}, \bibinfo{year}{1975}.
\newblock \bibinfo{title}{Some clarifications of the concept of a
  garden-of-eden configuration}.
\newblock \bibinfo{journal}{Journal of Computer and System Sciences}
  \bibinfo{volume}{10}, \bibinfo{pages}{77--82}.
%Type = Article
\bibitem[{Amoroso and Patt(1972)}]{1972_Amoroso}
\bibinfo{author}{Amoroso, S.}, \bibinfo{author}{Patt, Y.N.},
  \bibinfo{year}{1972}.
\newblock \bibinfo{title}{Decision procedures for surjectivity and injectivity
  of parallel maps for tessellation structures}.
\newblock \bibinfo{journal}{Journal of Computer and System Sciences}
  \bibinfo{volume}{6}, \bibinfo{pages}{448--464}.
%Type = Article
\bibitem[{Bardomero and Beck(2020)}]{bardomero2020frobenius}
\bibinfo{author}{Bardomero, L.}, \bibinfo{author}{Beck, M.},
  \bibinfo{year}{2020}.
\newblock \bibinfo{title}{Frobenius coin-exchange generating functions}.
\newblock \bibinfo{journal}{The American Mathematical Monthly}
  \bibinfo{volume}{127}, \bibinfo{pages}{308--315}.
%Type = Article
\bibitem[{Bhattacharjee and Das(2015)}]{2015_Bhattacharjee}
\bibinfo{author}{Bhattacharjee, K.}, \bibinfo{author}{Das, S.},
  \bibinfo{year}{2015}.
\newblock \bibinfo{title}{Reversibility of d-state finite cellular automata}.
\newblock \bibinfo{journal}{arXiv preprint arXiv:1502.01187} .
%Type = Article
\bibitem[{Bruckner(1979)}]{1979_Bruckner}
\bibinfo{author}{Bruckner, L.K.}, \bibinfo{year}{1979}.
\newblock \bibinfo{title}{On the garden-of-eden problem for one-dimensional
  cellular automata.}
\newblock \bibinfo{journal}{Acta Cybern.} \bibinfo{volume}{4},
  \bibinfo{pages}{259--262}.
\newblock \URLprefix
  \url{http://dblp.uni-trier.de/db/journals/actaC/actaC4.html#Bruckner80}.
%Type = Article
\bibitem[{Brun and Mlodinow(2020)}]{brun2020quantum}
\bibinfo{author}{Brun, T.A.}, \bibinfo{author}{Mlodinow, L.},
  \bibinfo{year}{2020}.
\newblock \bibinfo{title}{Quantum cellular automata and quantum field theory in
  two spatial dimensions}.
\newblock \bibinfo{journal}{Physical Review A} \bibinfo{volume}{102},
  \bibinfo{pages}{062222}.
%Type = Article
\bibitem[{Del~Rey and S{\'a}nchez(2009)}]{del2009reversibility}
\bibinfo{author}{Del~Rey, A.M.}, \bibinfo{author}{S{\'a}nchez, G.R.},
  \bibinfo{year}{2009}.
\newblock \bibinfo{title}{Reversibility of a symmetric linear cellular
  automata}.
\newblock \bibinfo{journal}{International Journal of Modern Physics C}
  \bibinfo{volume}{20}, \bibinfo{pages}{1081--1086}.
%Type = Article
\bibitem[{Du et~al.(2022)Du, Wang, Wang and Gao}]{2022_ChaoWang}
\bibinfo{author}{Du, X.}, \bibinfo{author}{Wang, C.}, \bibinfo{author}{Wang,
  T.}, \bibinfo{author}{Gao, Z.}, \bibinfo{year}{2022}.
\newblock \bibinfo{title}{Efficient methods with polynomial complexity to
  determine the reversibility of general 1d linear cellular automata over zp}.
\newblock \bibinfo{journal}{Information Sciences} \bibinfo{volume}{594},
  \bibinfo{pages}{163--176}.
\newblock \URLprefix
  \url{https://www.sciencedirect.com/science/article/pii/S0020025522000743},
  \DOIprefix\doi{https://doi.org/10.1016/j.ins.2022.01.045}.
%Type = Article
\bibitem[{Dzedzej et~al.(2020)Dzedzej, Wolnik, Nenca, Baetens and {De
  Baets}}]{DZEDZEJ2020104534}
\bibinfo{author}{Dzedzej, A.}, \bibinfo{author}{Wolnik, B.},
  \bibinfo{author}{Nenca, A.}, \bibinfo{author}{Baetens, J.M.},
  \bibinfo{author}{{De Baets}, B.}, \bibinfo{year}{2020}.
\newblock \bibinfo{title}{Efficient enumeration of three-state two-dimensional
  number-conserving cellular automata}.
\newblock \bibinfo{journal}{Information and Computation} \bibinfo{volume}{274},
  \bibinfo{pages}{104534}.
\newblock \URLprefix
  \url{https://www.sciencedirect.com/science/article/pii/S0890540120300225},
  \DOIprefix\doi{https://doi.org/10.1016/j.ic.2020.104534}.
  \bibinfo{note}{aUTOMATA 2017}.
%Type = Article
\bibitem[{Johnson(1975)}]{johnson1975finding}
\bibinfo{author}{Johnson, D.B.}, \bibinfo{year}{1975}.
\newblock \bibinfo{title}{Finding all the elementary circuits of a directed
  graph}.
\newblock \bibinfo{journal}{SIAM Journal on Computing} \bibinfo{volume}{4},
  \bibinfo{pages}{77--84}.
%Type = Inproceedings
\bibitem[{Jun(2009)}]{jun2009image}
\bibinfo{author}{Jun, J.}, \bibinfo{year}{2009}.
\newblock \bibinfo{title}{Image encryption method based on elementary cellular
  automata}, in: \bibinfo{booktitle}{IEEE Southeastcon 2009},
  \bibinfo{organization}{IEEE}. pp. \bibinfo{pages}{345--349}.
%Type = Incollection
\bibitem[{Kang et~al.(2008)Kang, Lee and Hong}]{kang2008pseudorandom}
\bibinfo{author}{Kang, B.H.}, \bibinfo{author}{Lee, D.H.},
  \bibinfo{author}{Hong, C.P.}, \bibinfo{year}{2008}.
\newblock \bibinfo{title}{Pseudorandom number generation using cellular
  automata}, in: \bibinfo{booktitle}{Novel algorithms and techniques in
  telecommunications, automation and industrial electronics}.
  \bibinfo{publisher}{Springer}, pp. \bibinfo{pages}{401--404}.
%Type = Article
\bibitem[{Kari(1990)}]{1990_Kari}
\bibinfo{author}{Kari, J.}, \bibinfo{year}{1990}.
\newblock \bibinfo{title}{Reversibility of 2d cellular automata is
  undecidable}.
\newblock \bibinfo{journal}{Physica D: Nonlinear Phenomena}
  \bibinfo{volume}{45}, \bibinfo{pages}{379--385}.
%Type = Article
\bibitem[{Kari(1994)}]{1994_Kari}
\bibinfo{author}{Kari, J.}, \bibinfo{year}{1994}.
\newblock \bibinfo{title}{Reversibility and surjectivity problems of cellular
  automata}.
\newblock \bibinfo{journal}{Journal of Computer and System Sciences}
  \bibinfo{volume}{48}, \bibinfo{pages}{149--182}.
%Type = Article
\bibitem[{Maiti et~al.(2010)Maiti, Ghosh, Munshi and
  Pal~Chaudhuri}]{2010_Maiti}
\bibinfo{author}{Maiti, N.S.}, \bibinfo{author}{Ghosh, S.},
  \bibinfo{author}{Munshi, S.}, \bibinfo{author}{Pal~Chaudhuri, P.},
  \bibinfo{year}{2010}.
\newblock \bibinfo{title}{Linear time algorithm for identifying the
  invertibility of null-boundary three neighborhood cellular automata}.
\newblock \bibinfo{journal}{Complex Systems} \bibinfo{volume}{19},
  \bibinfo{pages}{89}.
%Type = Article
\bibitem[{{Martı´n del Rey} and {Rodrı´guez
  Sánchez}(2011)}]{MARTINDELREY20118360}
\bibinfo{author}{{Martı´n del Rey}, A.}, \bibinfo{author}{{Rodrı´guez
  Sánchez}, G.}, \bibinfo{year}{2011}.
\newblock \bibinfo{title}{Reversibility of linear cellular automata}.
\newblock \bibinfo{journal}{Applied Mathematics and Computation}
  \bibinfo{volume}{217}, \bibinfo{pages}{8360--8366}.
\newblock \URLprefix
  \url{https://www.sciencedirect.com/science/article/pii/S0096300311003997},
  \DOIprefix\doi{https://doi.org/10.1016/j.amc.2011.03.033}.
%Type = Article
\bibitem[{Mohamed(2014)}]{mohamed2014parallel}
\bibinfo{author}{Mohamed, F.K.}, \bibinfo{year}{2014}.
\newblock \bibinfo{title}{A parallel block-based encryption schema for digital
  images using reversible cellular automata}.
\newblock \bibinfo{journal}{Engineering Science and Technology, an
  International Journal} \bibinfo{volume}{17}, \bibinfo{pages}{85--94}.
%Type = Inproceedings
\bibitem[{Moore(1962)}]{1962_Moore}
\bibinfo{author}{Moore, E.F.}, \bibinfo{year}{1962}.
\newblock \bibinfo{title}{Machine models of self-reproduction}, pp.
  \bibinfo{pages}{17--33}.
%Type = Inproceedings
\bibitem[{Myhill(1963)}]{1963_Myhill}
\bibinfo{author}{Myhill, J.R.}, \bibinfo{year}{1963}.
\newblock \bibinfo{title}{The converse of moore’s garden-of-eden theorem},
  pp. \bibinfo{pages}{685--686}.
%Type = Inproceedings
\bibitem[{von Neumann and Burks(1966)}]{1951_Neumann}
\bibinfo{author}{von Neumann, J.}, \bibinfo{author}{Burks, A.W.},
  \bibinfo{year}{1966}.
\newblock \bibinfo{title}{Theory of self reproducing automata}.
%Type = Book
\bibitem[{Rosin et~al.(2014)Rosin, Adamatzky and Sun}]{rosin2014cellular}
\bibinfo{author}{Rosin, P.}, \bibinfo{author}{Adamatzky, A.},
  \bibinfo{author}{Sun, X.}, \bibinfo{year}{2014}.
\newblock \bibinfo{title}{Cellular automata in image processing and geometry}.
\newblock \bibinfo{publisher}{Springer}.
%Type = Inproceedings
\bibitem[{Shallit(2008)}]{shallit2008frobenius}
\bibinfo{author}{Shallit, J.}, \bibinfo{year}{2008}.
\newblock \bibinfo{title}{The frobenius problem and its generalizations}, in:
  \bibinfo{booktitle}{International Conference on Developments in Language
  Theory}, \bibinfo{organization}{Springer}. pp. \bibinfo{pages}{72--83}.
%Type = Article
\bibitem[{Sutner(1991)}]{1991_Sutner}
\bibinfo{author}{Sutner, K.}, \bibinfo{year}{1991}.
\newblock \bibinfo{title}{De bruijn graphs and linear cellular automata}.
\newblock \bibinfo{journal}{Complex Syst.} \bibinfo{volume}{5},
  \bibinfo{pages}{19--30}.
%Type = Article
\bibitem[{Toffoli and Margolus(1990)}]{toffoli1990invertible}
\bibinfo{author}{Toffoli, T.}, \bibinfo{author}{Margolus, N.H.},
  \bibinfo{year}{1990}.
\newblock \bibinfo{title}{Invertible cellular automata: a review}.
\newblock \bibinfo{journal}{Physica D: Nonlinear Phenomena}
  \bibinfo{volume}{45}, \bibinfo{pages}{229--253}.
%Type = Article
\bibitem[{Wolnik and {De Baets}(2020)}]{WOLNIK2020180}
\bibinfo{author}{Wolnik, B.}, \bibinfo{author}{{De Baets}, B.},
  \bibinfo{year}{2020}.
\newblock \bibinfo{title}{Ternary reversible number-conserving cellular
  automata are trivial}.
\newblock \bibinfo{journal}{Information Sciences} \bibinfo{volume}{513},
  \bibinfo{pages}{180--189}.
\newblock \URLprefix
  \url{https://www.sciencedirect.com/science/article/pii/S0020025519310369},
  \DOIprefix\doi{https://doi.org/10.1016/j.ins.2019.10.068}.
%Type = Article
\bibitem[{Wolnik et~al.(2022)Wolnik, Dziemiańczuk, Dzedzej and {De
  Baets}}]{WOLNIK2022133075}
\bibinfo{author}{Wolnik, B.}, \bibinfo{author}{Dziemiańczuk, M.},
  \bibinfo{author}{Dzedzej, A.}, \bibinfo{author}{{De Baets}, B.},
  \bibinfo{year}{2022}.
\newblock \bibinfo{title}{Reversibility of number-conserving 1d cellular
  automata: Unlocking insights into the dynamics for larger state sets}.
\newblock \bibinfo{journal}{Physica D: Nonlinear Phenomena}
  \bibinfo{volume}{429}, \bibinfo{pages}{133075}.
\newblock \URLprefix
  \url{https://www.sciencedirect.com/science/article/pii/S0167278921002323},
  \DOIprefix\doi{https://doi.org/10.1016/j.physd.2021.133075}.
%Type = Article
\bibitem[{Yang et~al.(2015)Yang, Wang and Xiang}]{2015_ChaoWang}
\bibinfo{author}{Yang, B.}, \bibinfo{author}{Wang, C.}, \bibinfo{author}{Xiang,
  A.}, \bibinfo{year}{2015}.
\newblock \bibinfo{title}{Reversibility of general 1d linear cellular automata
  over the binary field z2 under null boundary conditions}.
\newblock \bibinfo{journal}{Information Sciences} \bibinfo{volume}{324},
  \bibinfo{pages}{23--31}.

\end{thebibliography}

\end{document}